\documentclass[10pt,conference]{IEEEtran}
\usepackage{graphicx}
\usepackage{amsfonts, amsmath, amssymb, amsthm, mathtools}
\allowdisplaybreaks
\usepackage{enumerate}
\usepackage{tabularx}
\usepackage{quiver}
\usepackage{physics}
\usepackage{setspace}
\usepackage[bb=boondox]{mathalfa}
\usepackage{tikz}
\usepackage{algorithm}
\usepackage{algpseudocode}
\usepackage{color}
\usepackage[hidelinks]{hyperref}
\usepackage{listings}
\usepackage{makecell}
\usepackage{fancyvrb}
\usepackage{listings}

\allowdisplaybreaks

\theoremstyle{definition}
\newtheorem{theorem}{Theorem}[section]
\newtheorem{definition}[theorem]{Definition}
\newtheorem{lemma}[theorem]{Lemma}
\newtheorem{prop}[theorem]{Proposition}
\newtheorem{corollary}[theorem]{Corollary}
\newtheorem{remark}[theorem]{Remark}

\DeclareMathOperator{\proj}{proj}

\DeclareMathOperator{\ad}{ad}

\DeclareMathOperator*{\argmin}{argmin}

\title{Stabilizers for Compiling Logical Circuits under Hardware Constraints}

\author{
  \IEEEauthorblockN{Jack Weinberg}
  \IEEEauthorblockA{Electrical and Computer Engineering\\
  University of Arizona\\
  Email: jackweinberg@arizona.edu}
  \and
  \IEEEauthorblockN{Narayanan Rengaswamy}
  \IEEEauthorblockA{Electrical and Computer Engineering\\
  University of Arizona\\
  Email: narayananr@arizona.edu}
}

\begin{document}

\bstctlcite{IEEEexample:BSTcontrol}

\maketitle

\begin{abstract}
    To implement quantum algorithms on a quantum computer, we must overcome the twin problems of fault-tolerance --- how can we realize a relatively noiseless computation by cleverly combining noisy components? --- and compilation --- how can we realize an arbitrary quantum algorithm given the basic operations available on the quantum device at hand? 
    We show how treating the former problem via error-correcting codes enables greater flexibility in resolving the latter.
    Specifically, we explicitly leverage the fact that error-correcting codes introduce redundancy which renders physically distinct operators logically indistinguishable.
    In terms of computation, it suffices to implement any operator logically equivalent to some target, yet from a compilation perspective, certain choices may be preferable to others.
    Our novel contribution is making this intuition precise in the general setting of the special unitary group.
    In particular, we describe how to reduce the problem of making a compilation-ideal choice to a least squares problem and provide a closed form solution thereof.
    Using our framework, it is possible to circumvent inserting costly swaps to adhere to hardware connectivity; instead, we could realize the logical target through a distinct physical Hamiltonian that is natively accessible.
    We elucidate our approach using the $[[4,2,2]]$ code.
    We discuss connections to compressed sensing that may pave the way to efficient compilation leveraging physical degrees of freedom.
\end{abstract}


\section{Introduction}
    Efficient quantum algorithms are thought to exist for solving problems that are out of reach for efficient probabilistic classical algorithms, i.e. it is conjectured that $\text{BPP} \subsetneq \text{BQP}$.
    For example, the quantum singular value transform (QSVT)~\cite{Gily_n_2019} encapsulates many quantum subroutines affording super-polynomial speedups over their best known classical counterparts, including quantum simulation~\cite{toyoizumi2023hamiltoniansimulationusingquantum} and methods for integer factorization and computing discrete logarithms~\cite{Shor_1997}. 
    Thus, there is substantial interest in physically realizing quantum computers in order to solve such problems which promise to have substantial impact in fields such as materials science \cite{wang2025quantumhybridmachinelearningmodels}, drug design \cite{smith2025bridgingquantumclassicalcomputing}, fundamental physics \cite{savage2025quantumsimulationsfundamentalphysics}, and cryptography \cite{mothukuri2026securingcryptographyagequantum}.
    However, current realizations are plagued by high intrinsic hardware error rates (see e.g. superconducting \cite{Google2024}, trapped ion \cite{Moses_2023}, or neutral atom \cite{reichardt2025faulttolerantquantumcomputationneutral} qubits).
    Thus, it appears that error correction, the technique of protecting quantum information via engineered redundancy, will be an essential component of practical quantum computation.
    If error correction is the theory of protecting quantum \emph{memory} from noise, then fault-tolerance is the theory of protecting quantum \emph{computation} from noise.
    Together, these theories elucidate how, in principle, one can combine faulty quantum components (with limited noise~\cite{Knill_1998,Knill_2000}) to realize \emph{encoded} quantum computation that is equivalent to but more reliable than its unencoded, or logical, counterpart.

    In addition to the challenges of error correction and fault-tolerance, anyone who aspires to build a quantum computer faces the challenge of \emph{compilation}: given readily accessible components (e.g., native gates via accessible Hamiltonians), how can one realize a target computation, either exactly or approximately?
    Approximate realization is the subject of \emph{unitary compilation} and is treated by the seminal work of Kitaev~\cite{kitaev_1997}, which was more recently refined by Kuperberg~\cite{kuperberg2023}.
    Exact realization is a subject of \emph{quantum optimal control} and is treated, for instance, by Chen \emph{et al.}~\cite{chen2022iterativegradientascentpulse} and by Lewis, Wiersema and Bose~\cite{lewis2025}.
    The typical limitation in accessibility of Hamiltonians is the underlying connectivity between qubits in the hardware, e.g., many superconducting qubit systems have only nearest neighbor connectivity or at most a few more longer range connections per qubit \cite{Linke2017}.
    The most common approach to address connectivity constraints is to insert swap gates that make the gates in the target circuit realizable. 
    However, this additional circuit depth is expensive and could potentially negate any quantum advantage provided by the logical algorithm.
    Hence, a cheaper alternative to swap-insertions is highly desirable.
    
    In the context of error-corrected quantum computation, \emph{it suffices to implement any operator whose action on the code space is identical to that of a logical target}; the action outside the code space is immaterial.
    Yet, from a compilation perspective, certain logically equivalent operators may be more readily implemented than others.
    In particular, certain choices may respect the connectivity constraints of the hardware at hand whereas others may not.
    In earlier work, Rengaswamy \emph{et al.}~\cite{Rengaswamy_2020} leveraged error correction to introduce the Logical Clifford Synthesis (LCS) algorithm, which efficiently and systematically explores all physical Clifford circuits that realize a target logical Clifford circuit, with the goal of minimizing circuit depth.
    Restricting to the Clifford group is advantageous since the Clifford group admits a symplectic representation of polynomial size in the number of physical qubits, as opposed to the unitary group which generally requires exponentially growing space.
    Though the algorithm itself only has quadratic complexity in the number of physical qubits, the solution space grows super-exponentially in the number of stabilizer generators for the code.
    This makes the search computationally challenging.
    Recently, Kuehnke \emph{et al.}~\cite{kuehnke2025hardwaretailoredlogicalcliffordcircuits} considered the same setting but with a focus on addressing limited hardware connectivity.
    The authors make use of the symplectic representation to reduce the problem of finding a connectivity-compliant physical realization of a logical operator to an \emph{integer quadratically constrained program} (IQCP).
    They demonstrate their method on several small codes with $\lessapprox 10$ physical qubits.
    Besides these, there are, to our knowledge, no other works that leverage physical degrees of freedom for compilation.
    
    Here, we work in greater generality, treating the connectivity-compliant logical operators problem in the setting of arbitrary special unitary operators with codes realized by any special unitary encoding map (recall that stabilizer codes are exactly those with Clifford encoding maps).
    Given that the special unitary group is obtained by restricting the global phases of the unitary group, which are physically irrelevant, we conclude that our method is essentially applicable to any unitary operator acting on the logical level of any error-correcting code.
    While the symplectic representation exploited by Kuehnke \emph{et al.} is not applicable in this more general setting, what becomes applicable, and indeed vital, is the Lie group structure of the special unitary group.
    
    Our primary contribution is to show that the above problem can be cast as an optimization problem over the space of \emph{stabilizer Hamiltonians} with a least squares objective function.
    Then, an immediate solution to this problem is the well-known Moore-Penrose pseudoinverse.
    We elucidate this approach via the $[[4,2,2]]$ code where a simple physical swap gate is known to realize the logical CNOT gate.
    If the swap is physically accessible on a hardware, then it is prudent to avoid a na{\"i}ve encoding of the CNOT and instead perform the swap.
    In our formulation, the pseudoinverse finds yet another alternative to the swap that shows the breadth of possibilities in a logically aware compilation framework, which is missing in conventional approaches.
    We think that ours is the first such framework that operates in full generality of the special unitary group, making it an important first step to navigate connectivity constraints in an error correction-aware setting.
    We also discuss the general computational burden of our approach and acknowledge that much more work must be done to make it amenable to real-time compilation.
    In this regard, we highlight a potential connection to compressed sensing by reformulating the problem as an $\ell_2$-$\ell_1$ minimization that encourages sparsity of the full physical Hamiltonian when expressed in the Pauli basis.
    We plan to delve deeper into this exciting direction in future work.

    The paper is organized as follows.
    Section \ref{sec:Technical-Background} discusses some technical background necessary to appreciate our results.
    Section \ref{sec:Theoretical-Results} describes our results.
    Section \ref{sec:Example} consists of a detailed example of our method in which we find a novel route for realizing logical $CNOT$ on the $[[4,2,2]]$ error-detecting code.
    We conclude with Section \ref{sec:Conclusion}.
    For the convenience of the reader, we provide pseudocode describing our least-squares strategy (see Algorithm \ref{alg:summary}).
    We assume access to a least-squares subroutine $\text{LS}$ which takes arguments $A$ and $b$ and returns $x$ minimizing $\|Ax-b\|_2$.
    Note that replacing $M$ with $J$ as in Definition \ref{def:Hamiltonian-cost} gives a method for minimizing the cost of a physical Hamiltonian with fixed logical action.
    
    \begin{algorithm}
    \hspace*{\algorithmicindent} \textbf{Input} Encoding Map $C : \mathbb{C}^{2^n} \to \mathbb{C}^{2^n}$ \\
    \hspace*{\algorithmicindent} \textbf{Input} Logical Hamiltonian $H_{\text{logical}} \in \mathfrak{su}(2^k)$ \\
    \hspace*{\algorithmicindent} \textbf{Input} Generators for accessible algebra\\ \hspace*{1.5cm} $g_1, \ldots, g_m \in \mathfrak{su}\left(2^n\right)$ \\
    \hspace*{\algorithmicindent} \textbf{Output} $H_{\text{corrected}}$
    \caption{Unitary Compilation}\label{alg:summary}
        \begin{algorithmic}[1]
            \Procedure{IM}{$C, U, g_1, ..., g_m$}
            \State $K \gets K$ as in Definition \ref{Def:K}
            \State $\mathcal{H} \gets \langle g_1, ..., g_m \rangle$
            \State $M \gets M$ as in Definition \ref{def:M}
            \State $A \gets A$ as in Definition \ref{def:inclusion}
            \State $H_{\text{Naïve}} \gets C K (H_{\text{logical}} \oplus 0_{2^n - 2^k}) K^\dagger C$
            \State $b \gets - M\left(H_{\text{Naive}} \right)$
            \State $H_{\text{Correction}} \gets C A \left( \text{LS}(M \circ A, b) \right) C^\dagger$
            \EndProcedure \\
            \Return $H_{\text{Correction}} + H_{\text{Naïve}}$
        \end{algorithmic}
    \end{algorithm}

\section{Technical Background}\label{sec:Technical-Background}
    \subsection{Lie Theory}\label{subsec:Lie-Theory}
    \subsubsection{Lie Groups}
    \begin{definition}
        A group $G$ equipped with a smooth structure is said to be a real Lie group if the two structures agree in the sense that group multiplication $G \times G \to G$ and inversion $G \to G$ are smooth~\cite[Chapter~7]{Lee2012}.
        Subsequently, all mentions of ``Lie groups" are to be understood as references to real Lie groups.
        That $G$ is a smooth manifold is to say, intuitively, that it locally admits a parametrization via real coordinates.
        That $G$ is a Lie group is to say that this parametrization admits a natural group structure.
        A homomorphism of Lie groups is a group homomorphism which is simultaneously smooth.
    \end{definition}

    \begin{definition}
        Define the special unitary group $SU(N)$ by
        \begin{multline}
            SU(N) := \{U \in GL(N, \mathbb{C}) : \\
            U U^\dagger = U^\dagger U = I, \det(U) = 1\}
        \end{multline}
        where $U^\dagger$ denotes the conjugate transpose of $U$ and the group operation is operator composition.
    \end{definition}

    \begin{prop}
        $SU(N)$ is a Lie group.
    \end{prop}

    \subsubsection{Lie Algebras}
    \begin{definition}
        A Lie algebra $\mathfrak{g}$ is a vector space over a field $\mathbb{F}$ equipped with a map $[\_, \_]: \mathfrak{g} \times \mathfrak{g} \to \mathfrak{g}$ satisfying the following for all $x,y,z \in \mathfrak{g}$ and $a,b \in \mathbb{F}$.
        \begin{enumerate}
            \item Bilinearity: 
            \begin{equation}
                [ax + by, z] = a[x,z] + b[y,z]
            \end{equation}
            \begin{equation}
                [z, ax + by] = a[z,x] + b[z,y]
            \end{equation}

            \item Alternativity:
            \begin{equation}
                [x,x] = 0
            \end{equation}
            
            \item The Jacobi identity:
            \begin{equation}
                [x, [y, z]] + [y, [z,x]] + [z, [x,y]] = 0
            \end{equation}
        \end{enumerate}
        
        We call $[\_, \_]$ a Lie bracket~\cite[Chapter~1]{Humphreys1975}.
        A Lie algebra homomorphism $T$ is a linear map respecting the Lie brackets of the domain and codomain in the sense that
        \begin{equation}
            T([x,y]) = [T(x), T(y)].
        \end{equation}
    \end{definition}

    \begin{definition}
        Define the special unitary lie algebra $\mathfrak{su}(N)$ by
        \begin{equation}
            \mathfrak{su}(N) := \{x \in M(N, \mathbb{C}) : \tr(x) = 0, x = - x^\dagger \},
        \end{equation}
        with the usual vector space operations over $\mathbb{R}$ (though not $\mathbb{C}$) and a Lie bracket given by the matrix commutator; note that $M(N, \mathbb{C})$ denotes the set of $N \times N$ complex matrices.
    \end{definition}

    \begin{prop}
        $\mathfrak{su}(N)$ is a Lie algebra.
    \end{prop}

    \begin{definition} We define the \emph{adjoint action} of a Lie algebra element $x \in \mathfrak{g}$ on the Lie algebra $\mathfrak{g}$ by
        \begin{equation}
        \begin{aligned}
            \ad_x \colon \mathfrak{g} & \longrightarrow \mathfrak{g} \\
            y & \longmapsto [x,y].
        \end{aligned}
        \end{equation}
    \end{definition}

    \begin{definition}
        We define the \emph{Killing form} on a Lie algebra over $\mathbb{F}$ by
        \begin{equation}
        \begin{aligned}
            \langle \_, \_ \rangle \colon \mathfrak{g} \times \mathfrak{g} & \longrightarrow \mathbb{F} \\
            (x,y) & \longmapsto \tr(\ad_x \circ \ad_y).
        \end{aligned}
        \end{equation}
    \end{definition}

    \begin{prop}
        On $\mathfrak{su}(N)$, the Killing form $\langle\_,\_\rangle_K$ satisfies
        \begin{equation}
            \langle x,y \rangle_K \propto \tr(xy).
        \end{equation}
        It is a negative-definite inner product~\cite[Page~583]{ASENS_1985_4_18_4_563_0}.
    \end{prop}

    \subsubsection{Lie Group-Algebra Correspondence}
    While they can be defined abstractly, Lie algebras also arise \emph{organically} as the tangent spaces of Lie groups at the identity or equivalently as left-invariant vector fields on Lie groups.
    More concretely, in the setting of matrix Lie groups (which we will exclusively deal with in this work), the Lie algebra of a Lie group is simply the pre-image of the Lie group under the exponential map where 
    \begin{equation}
        \exp(A) = \sum_{m=0}^\infty \frac{A^m}{m!}.
    \end{equation}
    Note that the exponential map as defined here is absolutely convergent.
    Thus, it raises no questions regarding well-definedness.
    The exponential map need not be surjective for arbitrary Lie groups, though it is for Lie groups such as $SU(N)$ which are topologically compact and connected~\cite{Tao_2015}.

    \begin{prop}
        $\mathfrak{su}(N)$ is the Lie algebra of $SU(N)$.
    \end{prop}
        
\section{Theoretical Results}\label{sec:Theoretical-Results}
    \subsection{Problem Setup}\label{subsec:Problem-Setup}
    Suppose Alice has a quantum computer with accessible Lie algebra $\mathcal{H}$, a $k$-qubit logical Hamiltonian $H_{\text{logical}}$ which she wishes to implement in an error-corrected fashion on her device, and a preferred $[[n,k,d]]$ code $\mathcal{C}$ realized by an encoding circuit $C \in SU(2^n)$.
    She may, however, encounter an obstacle: her logical Hamiltonian may not be \emph{na{\"i}vely implementable} in a sense which we will subsequently make precise.
    In the language of quantum optimal control, Alice's encoded Hamiltonian may not be accessible with respect to her device.
    Alice is not without recourse, though, for she knows that, in the context of error-corrected computation, it suffices to realize any Hamiltonian $H$ which is \emph{logically equivalent} to the na{\"i}ve Hamiltonian in the sense that the former and latter have the same action on the codespace.

    \begin{definition}\label{Def:K}
        We (non-uniquely) define an orthogonal linear map $K : \mathbb{C}^{2^n} \to \mathbb{C}^{2^n}$ by its action 
        \begin{equation}
            \ket{\psi} \oplus \vec{0}_{2^n - 2^k} \mapsto \ket{\psi} \otimes \ket{\vec{0}}_{2^{n-k}}
        \end{equation}
        for all $\ket{\psi} \in \mathbb{C}^{2^k}$ and by demanding that
        \begin{equation}
            K((\ket{\psi} \oplus \overline{0})^\perp) = K(\ket{\psi} \otimes \ket{\overline{0}})^\perp.
        \end{equation}
    \end{definition}

    \begin{remark}\label{remark:K}
        Our motivation for defining $K$ is twofold.
        First, we wish to transform Hamiltonians from a direct sum in which our stabilizer and logical subalgebras elements are most naturally expressed to the more familiar (in quantum information) tensor product on which the encoding map acts.
        Second, we demand orthogonality in particular so that the subsets on which the stabilizer and logical subalgebras act non-trivially remain disjoint.
        This property will be seen to imply that any element of the stabilizer subalgebra commutes with any element of the logical subalgebra.
    \end{remark}

    We have at hand two notions of orthogonality: one due to the Killing form on $\mathfrak{su}(n)$ and the other the ordinary inner product on the vector space in which quantum states live.
    Do these two notions coincide?
    In fact they do.

    \begin{lemma}\label{Killing-propto-HS}
        The Killing form $\langle \_, \_ \rangle_\text{K}$ on $\mathfrak{su}(N)$ is proportional to the Hilbert-Schmidt inner product $\langle \_, \_ \rangle_{\text{HS}}$.
    \end{lemma}
    \begin{proof}
        Since the elements of $\mathfrak{su}(N)$ are skew-Hermitian we have
        \begin{multline}
            \langle A,B \rangle_\text{K} \propto - 
            \tr(AB) = 
            - \tr(- A^\dagger B) = \\
            \tr(A^\dagger B) = 
            \langle A,B \rangle_{HS}.
        \end{multline}
    \end{proof}

    \begin{corollary}
        The Killing form and the Hilbert-Schmidt inner product induce identical notions of orthogonality.
    \end{corollary}
    \begin{proof}
        This follows immediately from Lemma \ref{Killing-propto-HS}.
    \end{proof}

    \begin{lemma}
        Let $\langle\_,\_\rangle$ be the ordinary inner product on $\mathbb{C}^{m \times m}$ and suppose $A,B \in \mathbb{C}^{m \times m}$.
        Then we have
        \begin{equation}
            \langle A,B \rangle_{\text{HS}} = \langle\text{vec}\left(A\right), \text{vec}\left(B\right)\rangle
        \end{equation}
    \end{lemma}
    \begin{proof}
        We have
        \begin{multline}
            \langle A,B \rangle_{\text{HS}} = 
            \tr(A^\dagger B) = \\
            \sum_{i=1}^m \sum_{j=1}^n \overline{A_{i,j}} B_{i,j} = 
            \langle\text{vec}\left(A\right), \text{vec}\left(B\right)\rangle
        \end{multline}
    \end{proof}

    Having clarified the relationship between our two notions of orthogonality, we turn to the task of explicitly describing the Hamiltonians which act purely on the codespace and its orthogonal complement.

    \begin{definition}\label{def:stabilizer-subalgebra}
        We call 
        \begin{equation}
            \mathfrak{su}(\mathcal{C^\perp}) := C K \left(0_{2^k} \oplus \mathfrak{su}(2^n - 2^k) \right) K^\dagger C^\dagger
        \end{equation}
        the \emph{stabilizer subalgebra} for $\mathcal{C}$ where $\oplus$ denotes the direct sum.
    \end{definition}

    \begin{remark}
        We note the remarkable correspondence between our definition of the stabilizer subalgebra and the form of the \emph{gauge group} due to Kuehnke \emph{et al.}~~\cite[Theorem~2]{kuehnke2025hardwaretailoredlogicalcliffordcircuits}.
        Namely, the form of their gauge matrix in Equation (12) is the Clifford unitary equivalent of our Hamiltonian perspective on the stabilizer subalgebra.
    \end{remark}

    \begin{definition}\label{def:logical-subalgebra}
        Let $C \in SU(2^n)$ be the encoding circuit for an $[[n,k,d]]$ code $\mathcal{C}$.
        We call 
        \begin{equation}
            \mathfrak{su}(\mathcal{C}) := CK\left(\mathfrak{su}(2^k) \oplus 0_{2^n-2^k} \right)K^\dagger C^\dagger
        \end{equation}
        the \emph{logical subalgebra} where $\oplus$ denotes the direct sum.
    \end{definition}

    \begin{remark}
        A well-known identity from Lie theory states that for matrix Lie groups $G$, we have
        \begin{equation}
            \exp(g A g^{-1}) = g \exp(A) g^{-1}
        \end{equation}
        for $g \in G$ and $A \in \mathfrak{g}$ the Lie algebra of $G$.
        In the context of error-corrected quantum computation, this is to say that conjugating a Hamiltonian by an encoding map and then exponentiating is strictly equivalent to first exponentiating and then conjugating.
        Thus, notions of logical and stabilizer unitary operators induce satisfactory notions of logical and stabilizer Hamiltonians.
    \end{remark}

    \begin{remark}
        As promised, we now explicate what we mean when we say that Alice's logical Hamiltonian is or is not na{\"i}vely implementable on her device with accessible Lie subalgebra $\mathcal{H}$.
        As in Definition \ref{def:logical-subalgebra}, a $k$-qubit logical Hamiltonian on a code with encoding map $C$ has a physical form
        \begin{equation}
            C K\left(H_{\text{logical}} \oplus 0_{2^n - 2^k}\right) K^\dagger C^\dagger.
        \end{equation}
        This is the \emph{na{\"i}ve encoding}.
        We thus say $H_{\text{logical}}$ is \emph{na{\"i}vely accessible} if and only if
        \begin{equation}
            C K \left(H_{\text{logical}} \oplus 0_{2^n - 2^k}\right) K^\dagger C^\dagger \in \mathcal{H}.
        \end{equation}
    \end{remark}

    \begin{lemma}\label{lemma:k-correct}
        The naïve implementation of $H_{\text{logical}}$ is in fact an implementation of $H_{\text{logical}}$, i.e. it has the desired logical action on the codespace.
    \end{lemma}
    \begin{proof}
        Note that
        \begin{multline}
            \exp\left(K\left(H_{\text{logical}} \oplus 0_{2^n - 2^k}\right) K^\dagger \right) \ket{\psi} \otimes \ket{\vec{0}}_{2^{n-k}} = \\
            K \exp \left(H_{\text{logical}} \oplus 0_{2^n - 2^k}\right) \ket{\psi} \oplus \vec{0}_{2^n - 2^k} = \\
            K \sum_{m=0}^\infty \frac{1}{m!}(H_{\text{logical}} \oplus 0_{2^n - 2^k})^m \ket{\psi} \oplus \vec{0}_{2^n - 2^k} = \\
            K \sum_{m=0}^\infty \frac{1}{m!}(H_{\text{logical}}^m \oplus \left(0_{2^n - 2^k}\right)^m) \ket{\psi} \oplus \vec{0}_{2^n - 2^k} = \\
            K \left(\exp\left(H_{\text{logical}}\right) \oplus \exp(0_{2^n - 2^k}) \right) \ket{\psi} \oplus \vec{0}_{2^n - 2^k} = \\
            K \left(\exp\left(H_{\text{logical}}\right) \ket{\psi} \oplus \exp(0_{2^n - 2^k}) \vec{0}_{2^n - 2^k} \right) = \\
            K \left( \exp\left(H_{\text{logical}}\right) \ket{\psi} \oplus \vec{0}_{2^n - 2^k}\right) = \\
            \exp\left(H_{\text{logical}}\right) \ket{\psi} \otimes \ket{\vec{0}}_{2^{n-k}}
        \end{multline}
        where we used the fact that
        \begin{equation}
            (A \oplus B)(C \oplus D) = AC \oplus BD.
        \end{equation}
        Then we have
        \begin{multline}
            \exp\left(C K\left(H_{\text{logical}} \oplus 0_{2^n - 2^k}\right) K^\dagger C^\dagger\right) \\
            C\left(\ket{\psi} \otimes \ket{\vec{0}}_{2^{n-k}}\right)C^\dagger = \\
            C \exp\left(K\left(H_{\text{logical}} \oplus 0_{2^n - 2^k}\right) K^\dagger \right) C^\dagger \\
            C\left(\ket{\psi} \otimes \ket{\vec{0}}_{2^{n-k}}\right)C^\dagger = \\
            C \exp\left(K\left(H_{\text{logical}} \oplus 0_{2^n - 2^k}\right) K^\dagger \right)\\
            \left(\ket{\psi} \otimes \ket{\vec{0}}_{2^{n-k}}\right)C^\dagger = \\
            C \left(\exp\left(H_{\text{logical}}\right) \ket{\psi} \otimes \ket{\vec{0}}_{2^{n-k}}\right) C
        \end{multline}
    \end{proof}

    \begin{lemma}\label{commutation}
        Fix a code $\mathcal{C}$.
        Then the elements of the logical subalgebra for $\mathcal{C}$ commute with the elements of the stabilizer subalgebra for $\mathcal{C}$.
    \end{lemma}
    \begin{proof}
        Note that $A \oplus 0$ commutes with $0 \oplus B$ since these operators act non-trivially on disjoint subsets of $\mathbb{C}^{2^n}$.
        Up to a change of basis, this is exactly the content of Lemma \ref{commutation}.
    \end{proof}
    
    \begin{theorem}\label{thm:stabilizer-equivalence}
        Suppose $H_{\text{logical}} \in \mathfrak{su}(2^k)$ and $H \in \mathfrak{su}\left(2^n\right)$ is the encoded form of some logical operator in the sense that $H$ is block diagonal with respect to the decomposition $\mathbb{C}^{2^n} = \mathcal{C} \oplus \mathcal{C}^\perp$.
        Let $C \in SU(2^n)$ be the encoding circuit for an $[[n,k,d]]$ code $\mathcal{C}$.
        Then $\exp(CK (H_{\text{logical}} \oplus 0_{2^n - 2^k})K^\dagger C^\dagger)$ and $\exp(H)$ have the same action on $\mathcal{C}$ if and only if there exists $H_{\text{stabilizer}} \in \mathfrak{su}(2^n - 2^k)$ such that
        \begin{multline}
            H = CK \left(H_{\text{logical}} \oplus 0_{2^n - 2^k}\right) K^\dagger C^\dagger + \\
            C K \left(0_{2^k} \oplus H_{\text{stabilizer}}\right) K^\dagger C^\dagger.
        \end{multline}
    \end{theorem}

    \begin{proof}
        Observe that the left-hand summand is the only one with non-trivial action on the codespace.
        This implies sufficiency.
        To see necessity, note that by hypothesis, $H$ can be written as a sum of a stabilizer and logical part.
    \end{proof}

    \begin{remark}
        Theorem \ref{thm:stabilizer-equivalence} demonstrates that the ``physical" degrees of freedom, that is, the equivalence class of special unitary operators logically equivalent to a distinguished logical special unitary operator are parametrized by the elements of the stabilizer subalgebra.
        That is, any unitary with trivial action on the codespace has non-empty image in the stabilizer subalgebra under the exponential map.
        We have not, however, parametrized all degrees of freedom prior to exponentiation, i.e. the pre-image under the exponential map of the set of stabilizer unitary operators is not equal to the stabilizer subalgebra. This pre-image is not a vector space at all.
    \end{remark}

    \begin{lemma}\label{maximal-subspace}
        The stabilizer subalgebra is the maximal vector space contained within the pre-image under the exponential map of the set of unitary operators which act trivially on the codespace. 
    \end{lemma}
    \begin{proof}
        Consider $U \in SU(2^n)$.
        With respect to the decomposition $\mathbb{C}^{2^n} = \mathcal{C} \oplus \mathcal{C}^\perp$, we have a block representation:
        \begin{equation}
            U = 
            \begin{bmatrix}
                A & B \\
                C & D
            \end{bmatrix}
        \end{equation}
        By hypothesis, $A = I$ and $C = 0$.
        Unitarity then forces $B = 0$.
        Therefore, the set of unitary operators which act trivially on the codespace correspond to the possible values of $D$, i.e. to $SU(2^n - 2^k).$
        Now, let $L$ be the vector space of block matrices in $\mathfrak{su}\left(2^N\right)$ under the decomposition $\mathbb{C}^{2^n} = \mathcal{C} \oplus \mathcal{C}^\perp$ with only the lower-right block possibly nonvanishing and let $\mathcal{D}$ be the set of block diagonal special unitary operators under the same decomposition.
        It suffices to check that if $x \in W \subseteq \mathcal{D}$ for $W$ a vector space then $x \in L$; we can conjugate by $CK$ to recover the familiar notion of the stabilizer subalgebra.
        If $x \in W \subseteq \mathcal{D}$ for $W$ a vector space then in particular for all $\epsilon > 0$ we have $\epsilon x \in W$.
        Recall that the exponential map is a diffeomorphism in some neighborhood of the identity \cite[Corollary~3.44]{Hall:371445}.
        It is thus a bijection therein in particular.
        Therefore, for sufficiently small $\epsilon$, we have that
        \begin{equation}
            \exp^{-1} \left( \exp(\epsilon X) \right)
        \end{equation}
        is a singleton.
        Moreover, this pre-image is non-empty in $L$ by the surjectivity of the exponential map for $SU(N)$.
        Therefore, $\epsilon X \in L$.
        Since $L$ is manifestly a vector space we have $\frac{1}{\epsilon} \epsilon X = X \in L$.
    \end{proof}

    \begin{remark}
        The operational implication of Lemma \ref{maximal-subspace} is that if we demand a linear structure on the solution space of logically trivial Hamiltonians then the stabilizer subalgebra furnishes us with the most general approach.
    \end{remark}

    \subsection{Device Compatibility}\label{subsec:Device-Compatibility}
    The operational implication of Theorem \ref{thm:stabilizer-equivalence} is that in order to choose among the physical Hamiltonians logically equivalent to her target with respect to $\mathcal{C}$, it suffices for Alice to choose an element of the stabilizer subalgebra of $\mathcal{C}$.
    How should she choose a stabilizer in such a way that compatibility with her device is maximized?
    We first need to formalize a notion of compatibility.
    We consider the setting in which the experimenter has access to a subalgebra $\mathcal{H}$ of freely available Hamiltonians and no access to elements outside $\mathcal{H}$.
    Given that our goal is to choose an element of the stabilizer subalgebra such that the \emph{total physical Hamiltonian} is accessible, we use the following cost function:
    \begin{multline}\label{stabilizer-minimization-problem}
        \argmin_{H_{\text{stabilizer}} \in \mathfrak{su}(2^n - 2^k)} \\
        \|\proj_\mathcal{H}(CK\left(H_{\text{logical}} \oplus 0_{2^n - 2^k}\right)K^\dagger C^\dagger + \\
        C K \left(0_{2^k} \oplus H_{\text{stabilizer}}\right) K^\dagger C^\dagger) - \\
        (C K \left(H_{\text{logical}} \oplus 0_{2^n - 2^k}\right) K^\dagger C^\dagger + \\
        CK \left(0_{2^k} \oplus H_{\text{stabilizer}}\right)K^\dagger  C^\dagger)\|^2.
    \end{multline}
    In fact, a different choice of cost function will have leave our core result (Theorem \ref{thm:projection-minimizer}) essentially unchanged provided the cost function is a $2$-norm of an affine function of the argument $H_{\text{stabilizer}}$.
        
    \subsection{Optimal Choice of Stabilizer}\label{subsec:Optimal-Choice-Of-Stabilizer}
    We claim that (\ref{stabilizer-minimization-problem}) has the form of a least-squares problem.
    \begin{definition}\label{def:M} Define
    \begin{equation}
    \begin{aligned}
        M \colon \mathfrak{su}\left(2^n\right) & \longrightarrow \mathfrak{su}\left(2^n\right) \\
        v & \longmapsto \proj_\mathcal{H} (C v C^\dagger) - C v C^\dagger.
    \end{aligned}
    \end{equation}
    \end{definition}
    \begin{lemma}
        M so defined is linear.
    \end{lemma}
    \begin{proof}
        Recall that linear endomorphisms of a vector space form a ring under the usual operations.
    \end{proof}

    \begin{definition}\label{def:inclusion}
        We define
        \begin{equation}
        \begin{aligned}
        A : \mathfrak{su}(2^n - 2^k) & \longrightarrow \mathfrak{su}(\mathcal{C}^\perp) \\
        v & \longmapsto K (0_{2^k} \oplus v)K^\dagger 
    \end{aligned}
    \end{equation}
    \end{definition}

    \begin{theorem}\label{thm:projection-minimizer}
        The choice of stabilizer minimizing (\ref{stabilizer-minimization-problem}) is given by
        \begin{equation}
            - (M \circ A)^+ M \left(K (H_{\text{logical}} \oplus 0_{2^n - 2^k})K^\dagger \right)
        \end{equation}
        where $T^+$ denotes the Moore-Penrose Pseudoinverse of $T$.
    \end{theorem}
    \begin{proof}
        We have
        \begin{multline}
            \proj_\mathcal{H}\bigg(CK \left(H_{\text{logical}} \oplus 0_{2^n - 2^k}\right)K^\dagger C^\dagger + \\
            CK \left(0_{2^k} \oplus H_{\text{stabilizer}}\right) K^\dagger C^\dagger\bigg) - \\
            \bigg(CK \left(H_{\text{logical}} \oplus 0_{2^n - 2^k}\right)K^\dagger C^\dagger + \\
            CK \left(0_{2^k} \oplus H_{\text{stabilizer}}\right) K^\dagger C^\dagger\bigg) = \\
            \proj_\mathcal{H}\bigg(CK \left(H_{\text{logical}} \oplus 0_{2^n-2^k}\right)K^\dagger C^\dagger\bigg) + \\
            \proj_\mathcal{H}\bigg(CK \left(0_{2^k} \oplus H_{\text{stabilizer}}\right) K^\dagger C^\dagger\bigg) - \\
            \bigg(C K \left(H_{\text{logical}} \oplus 0_{2^n - 2^k}\right)K^\dagger C^\dagger  + \\
            CK \left(0_{2^k} \oplus H_{\text{stabilizer}}\right) K^\dagger C^\dagger\bigg) = \\
            \proj_\mathcal{H}\bigg(CK\left(H_{\text{logical}} \oplus 0_{2^n - 2^k}\right)K^\dagger C^\dagger\bigg) - \\
            \bigg(CK\left(H_{\text{logical}} \oplus 0_{2^n - 2^k}\right)K^\dagger C^\dagger \bigg) + \\
            \proj_\mathcal{H}\bigg(CK \left(0_{2^k} \oplus H_{\text{stabilizer}}\right) K^\dagger C^\dagger \bigg)  - \\
            \bigg(CK \left(0_{2^k} \oplus H_{\text{stabilizer}}\right) K^\dagger C^\dagger\bigg) = \\
            M(K (H_{\text{logical}} \oplus 0_{2^n - 2^k})K^\dagger ) + \\
            M(K (0_{2^k} \oplus H_{\text{stabilizer}})K^\dagger )
        \end{multline}
        hence (\ref{stabilizer-minimization-problem}) becomes
        \begin{multline}\label{simplified-stabilizer-minimization-problem}
            \argmin_{H_{\text{stabilizer}} \in \mathfrak{su}(2^n - 2^k)} \|M \circ A(H_{\text{stabilizer}}) + \\
            M\left(K (H_{\text{logical}} \oplus 0_{2^n - 2^k})K^\dagger \right)\|^2.
        \end{multline}
        If we set $b = - M\left(K (H_{\text{logical}} \oplus 0_{2^n - 2^k})K^\dagger \right)$ then we can see that (\ref{simplified-stabilizer-minimization-problem}) has exactly the usual form
        \begin{equation}
            \argmin_{u \in \mathfrak{su}(2^n - 2^k)} \|(M \circ A)(u) - b\|^2
        \end{equation}
        of a least-squares problem.
        In turn, we may solve it using the Moore-Penrose pseudoinverse of $M \circ A$~\cite[Theorem~6.1]{Barata_2011}.
        We thus have a minimizer
        \begin{equation}
            - (M \circ A)^+ \circ M \big(K (H_{\text{logical}} \oplus 0_{2^n - 2^k})K^\dagger \big).
        \end{equation}
    \end{proof}
    \begin{remark}
        The final physical Hamiltonian we wish to implement is exactly the sum of the na{\"i}ve physical Hamiltonian and the optimal stabilizer correction.
    \end{remark}

    \subsection{Correctibly Accessible Hamiltonians}\label{subsec:Correctibly-Accessible-Hamiltonians}
    \begin{lemma}\label{correctibility-criterion}
        An encoded Hamiltonian is correctible to an accessible one on a code $\mathcal{C}$ if and only if it is an element of
        \begin{equation}
            \mathcal{H} + \mathfrak{su}(\mathcal{C}^\perp).
        \end{equation}
    \end{lemma}
    \begin{proof}
        The content of this lemma is exactly that an encoded Hamiltonian is correctibly accessible if and only if it differs from an accessible one by a stabilizer. 
        This follows by construction.
    \end{proof}
    \begin{remark}
        We have
        \begin{equation}
            \mathfrak{su}(\mathcal{C}^\perp) \not\subseteq \mathcal{H} \implies 
            \mathcal{H} \subsetneq \mathcal{H} + \mathfrak{su}(\mathcal{C}^\perp)
        \end{equation}
        hence our stabilizer correction extends the space of accessible Hamiltonians \emph{in our logical sense} in the case where
        \begin{equation}
            CK \left(0_{2^k} \oplus \mathfrak{su}(2^n - 2^k)\right)K^\dagger C^\dagger \not\subseteq \mathcal{H}
        \end{equation}
        Note that our procedure returns the sum of a stabilizer and logical Hamiltonian. These commute, so we can exponentiate them independently and take the product in order to find the exponential of their sum. Independently, they both preserve the codespace by construction, so their product will as well.
    \end{remark}

    \begin{lemma}
        Fix an $[[n,k,d]]$ code $\mathcal{C}$.
        Then the space of correctibly accessible \emph{logical} encoded Hamiltonians is exactly
        \begin{equation}
            \mathfrak{su}(\mathcal{C}) \cap \left(\mathcal{H} + \mathfrak{su}(\mathcal{C}^\perp)\right).
        \end{equation}
    \end{lemma}
    \begin{proof}
        This follows immediately from Lemma \ref{correctibility-criterion}.
    \end{proof}

    \subsection{Approximation Error}
    \begin{theorem}\label{thm:unitary-error}
        Let $V = \proj_\mathcal{H} H$.
        Then
        \begin{equation}
            \|\exp(H) - \exp(V)\| \leq 
            \|H-V\| \exp(\|H\|)
        \end{equation}
        for any sub-multiplicative norm.
    \end{theorem}
    \begin{proof}
        Recall that for linear operators $A, B$ we have
        \begin{equation}
            \|\exp(A) - \exp(B)\| \leq \|A-B\| \exp(\max(\|A\|, \|B\|))
        \end{equation}
        for any sub-multiplicative norm.
        Note that
        \begin{equation}
            \max(\|H\|, \|V\|) = \|H\|.
        \end{equation}
    \end{proof}

    \begin{remark}
        The operational implication of Theorem \ref{thm:unitary-error} is that the error incurred by approximating a Hamiltonian $H$ by the nearest accessible Hamiltonian can be bounded in terms of the operator norm of the difference and the norm of $H$.
        In particular, one sees that increasing $\dim \mathcal{H}$ can never increase $\|H - V\|$ and thus can never increase the error incurred by this approximation.
    \end{remark}

    \subsection{Naïve Time Complexity of Our Procedure}\label{subsec:Naïve-Time-Complexity}
    \begin{table}
    \label{tab:complexity}
    \centering
    \caption{Steps for naïvely computing optimal stabilizer correction and their complexity}
    \begin{tabular}{|c|c|}
        \hline
        Step & Complexity \\
        \hline
        \thead{Direct Sum of $2^k \times 2^k$ and \\ $(2^n -2^k)\times (2^n - 2^k)$ Matrices} & \thead{$O(2^{2n})$} \\\hline
        \thead{Computing $K^\dagger$} & \thead{$O(2^{2n})$} \\\hline
        \thead{Conjugation by $K$} & \thead{$O(2^{\omega n})$} \\\hline
        \thead{Computation of Matrix \\ Representation of $M$} & \thead{$O(2^{4n})$} \\\hline
        \thead{Application of $M$} & \thead{$O(2^{4n})$} \\\hline
        \thead{Computation of $(M \circ A)^+$} & \thead{$O((2^n-2^k)^4 \cdot 2^{2n})$} \\\hline
        \thead{Application of $(M \circ A)^+$} & \thead{$O((2^n-2^k)^2 \cdot 2^{2n})$} \\\hline
        \thead{Application of $A$} & \thead{$O(2^n-2^k)^2 \cdot 2^{2n})$} \\
        \hline
    \end{tabular}
    \end{table}
    
    Inspecting Theorem \ref{thm:projection-minimizer} reveals that the operations in Table \ref{tab:complexity} are sufficient to compute an optimal stabilizer correction.
        
    For fixed $k$, the calculation of $(M \circ A)^+$ is the dominating step with time complexity $O(2^{6n})$ where $n$ is the number of physical qubits.
    Since an input unitary or Hamiltonian on $n$ qubits takes space proportional to $2^{2n}-1$ bits to store, we conclude that our method is time-efficient in a complexity-theoretic sense though it scales exponentially in the number of physical qubits for a fixed number of logical qubits.

    \subsection{Improved Time Complexity of Our Procedure}\label{subsec:Improved-Time-Complexity}
    \begin{theorem}\label{thm:improved-complexity}
        There exists an algorithm which finds the optimal stabilizer correction in $O\left(2^{(\omega + 2)n}\right)$ where $\omega$ is the matrix multiplication constant.
    \end{theorem}
    \begin{proof}
        We provide a constructive proof.
        \emph{A priori}, we know that the nullspace of $M$ is $C^\dagger \mathcal{H} C$ hence the nullspace of $M \circ A$ is exactly $A^{-1} (C^\dagger \mathcal{H} C)$.
        Moreover, we have that the elements of $(C^\dagger \mathcal{H} C)^\perp$ transform like $v \mapsto -C v C^\dagger$ under the action of $M$.
        
        Our strategy, then, is twofold:
        First, we construct an orthonormal basis for $A^{-1} (C^\dagger \mathcal{H} C)$. 
        Then, we extend this to an orthonormal basis for $\text{dom } A = \mathbb{C}^{(2^n - 2^k) \times (2^n - 2^k)}$.
        
        We begin by finding a basis $\mathcal{B}_2$ for $C^\dagger \mathcal{H} C$.
        This is simply the set of $\dim \mathcal{H}$ basis vectors for $\mathcal{H}$, which we are given, acted upon via conjugation by $C^\dagger$ for an overall complexity of $O(\dim \mathcal{H} \: 2^{\omega n})$.
        We can then find $\mathcal{B}_1$, a basis for $A^{-1} (C^\dagger \mathcal{H} C)$, as follows:
        For each $v \in \mathcal{B}_2$, if $K v K^\dagger$ has an upper diagonal $0_{2^n - 2^k}$ block then $v$ is in the image of $A$.
        Moreover, its pre-image is the lower diagonal block of $K v K^\dagger$ and can thus be added to $\mathcal{B}_1$.
        If the upper diagonal block of $K v K^\dagger$ is non-zero then $v$ is not in the image of $A$.
        It suffices to apply this procedure to each of the $\dim \mathcal{H}$ basis vectors, each at a cost of $O(2^{2n} + 2 \times 2^{\omega n}) = O(2^{\omega n})$ for a total cost of $O(\dim \mathcal{H} \: 2^{\omega n})$.
        
        In the aforementioned basis, we have that
        \begin{equation}
            (M \circ A) = 
            \begin{bmatrix}
            0 \\
            -G_C
            \end{bmatrix}
        \end{equation}
        where $G_C$ takes $v$ to $C v C^\dagger$.
        Then, since $G_C$ is invertible, $(M \circ A)$ has linearly independent columns, hence we have
        \begin{multline}
            (M \circ A)^+ = ((M \circ A)^\dagger (M \circ A))^{-1} (M \circ A)^\dagger = \\
            \left(\begin{bmatrix}
            0 & -G_{C^\dagger}
            \end{bmatrix}
            \begin{bmatrix}
            0 \\
            -G_C
            \end{bmatrix}
            \right)^{-1}
            \begin{bmatrix}
            0 & -G_{C^\dagger}
            \end{bmatrix} = \\
            I \begin{bmatrix}
            0 & -G_{C^\dagger}
            \end{bmatrix} = 
            \begin{bmatrix}
            0 & -G_{C^\dagger}
            \end{bmatrix}.
        \end{multline}
        
        We can find $G_{C^\dagger}$ by:
        \begin{enumerate}[(a)]
            \item Finding $C^\dagger$ in $O(2^{2n})$.
            \item Applying $v \mapsto - C^\dagger v C$ a total of $\dim \mathcal{H}^\perp$ times, each with a cost of $O(2 \times 2^{\omega n}) = O(2^{\omega n})$ where $\omega$ is the matrix multiplication constant, for a total complexity of $O(\dim(\mathcal{H}^\perp) \: 2^{\omega n})$.
        \end{enumerate}
        Since $\omega \geq 2$ the overall complexity of finding the block pseudoinverse is $O(\dim(\mathcal{H}^\perp) \: 2^{\omega n})$.
                
        Note that 
        \begin{equation}
            \dim(\mathcal{H}) + \dim(\mathcal{H}^\perp) = 2^{2n} - 1
        \end{equation}
        hence our new method is
        \begin{equation}
            O(\dim(\mathcal{H}) \: 2^{\omega n} + \dim(\mathcal{H}^\perp) \: 2^{\omega n}) = O\left(2^{(\omega + 2)n}\right).
        \end{equation}
    \end{proof}

    \begin{remark}
        We have yet to make use of parallel strategies for computing the Moore-Penrose pseudoinverse.
        Given the procedure of Theorem \ref{thm:improved-complexity}, a natural next step is to parallelize the application of the map $v \mapsto C v C^\dagger$ for various $v$.
        We may likewise parallelize the procedure for finding $\mathcal{B}_1$.
        Doing so gives a parallel procedure with $O(2^{\omega n})$ time complexity and $O(\max(\dim(\mathcal{H}), \dim(\mathcal{H}^\perp)))$ width.
    \end{remark}

    \subsection{Regularization}
    \label{subsec:Regularization}
    
    Suppose that in addition to valuing the accessibility of a solution to our optimization problem, as we have already done, we also, for the benefit of the experimentalist, valued sparsity in a particular basis $\mathcal{B}$ for $\mathfrak{su}\left(2^n\right)$.
    In that basis, for $\lambda \in [0, \infty)$ we have a minimization problem:
    \begin{multline}\label{l1-l2 minimization}
        \argmin_{H_{\text{stabilizer}} \in \mathfrak{su}(2^n - 2^k)} \|M \circ A(H_{\text{stabilizer}}) + \\
        M\left(K (H_{\text{logical}} \oplus 0_{2^n - 2^k})K^\dagger \right)\|_2^2 + \\
        \lambda \|C A(H_{\text{stabilizer}}) C^\dagger + \\
        C\left(K (H_{\text{logical}} \oplus 0_{2^n - 2^k})K^\dagger \right)C^\dagger\|_1.
    \end{multline}
    Equation $(\ref{l1-l2 minimization})$ has the form
    \begin{equation}
        \min_x f(x) + g(x)
    \end{equation}
    subject to
    \begin{equation}
        Ix - Ix = 0
    \end{equation}
    with $f,g$ convex and is thus amenable to the alternating direction method of multipliers (ADMM) \cite{boyd_2010}.

    \subsection{Weighted Hamiltonians}
    \label{subsec:Weighted-Hamiltonians}
    Thus far, we have declared Hamiltonians to either be completely accessible or completely inaccessible.
    This can be understood as a binary notion of cost.
    In this section, we study a generalized notion of cost wherein all Hamiltonians have associated weights: higher weights denote higher costs.
    This is motivated by e.g. the setting of a quantum computer comprised of a reconfigurable atom array: while the operation of switching atoms in the array is technically accessible, it is, in some sense, more costly than interactions between atoms which are already adjacent.
    While we have not addressed this subtlety thus far, we claim that our framework can be modified in a straightforward way in order to capture it.
    Suppose Alice can implement any physical Hamiltonian, though with varying costs.
    We define a function $J$ capturing the extent to which a Hamiltonian solution uses ``expensive'' accessible Hamiltonians $H_j$ - furnishing a basis for $\mathfrak{su}(2^n)$ - with weights $w_j$ for $j$ in index set $\mathcal{J}$:
    \begin{definition}\label{def:Hamiltonian-cost}
        \begin{equation}
        \begin{aligned}
            J \colon \mathbb{C}^{2^n \times 2^n} & \longrightarrow \mathbb{C}^{2^n \times 2^n} \\
            \gamma & \longmapsto \sum_{j \in \mathcal{J}} w_j \proj_{H_j} \gamma
        \end{aligned}
        \end{equation}
    \end{definition}
    Our cost function becomes:
    \begin{multline}
        \|J(C A(H_{\text{stabilizer}}) C^\dagger + \\
        C\left(K (H_{\text{logical}} \oplus 0_{2^n - 2^k})K^\dagger \right)C^\dagger)\|^2_2.
    \end{multline}
    By the argument in Section \ref{subsec:Optimal-Choice-Of-Stabilizer}, this too can be minimized by application of a pseudoinverse.

\section{CNOT on Error-Detecting Code}\label{sec:Example}
    We investigate implementations of CNOT on the $[[4,2,2]]$ error-detecting code using a toy and a non-trivial example of an accessible subalgebra.
    Recall that $K$ as in Definition \ref{Def:K} is not uniquely defined.
    For concreteness, we define $K$ uniquely as follows:
    \begin{equation}
        K\left(\ket{\psi} \oplus \vec{0}_{2^n - 2^k}\right) = \ket{\psi} \otimes \ket{\vec{0}}_{2^{n-k}}
    \end{equation}
    and
    \begin{equation}
        K\left(\ket{2^k + m}\right) = \ket{i_m} \otimes \ket{j_m}
    \end{equation}
    for $m = 0, ..., 2^n - 2^k - 1$ where
    \begin{equation}
        i_m = \left\lfloor \frac{m}{2^{n-k}-1} \right\rfloor
    \end{equation}
    and
    \begin{equation}
        j_m = \left(m \mod \left(2^{n-k}-1\right)\right) + 1.
    \end{equation}
    \subsection{Toy Example}
    Our goal in this subsection is to engineer the accessible subalgebra such that the well-known SWAP Hamiltonian is recovered as an implementation of CNOT on the $[[4,2,2]]$ code.

    \begin{figure}[t]
    \centering
        \scalebox{0.9}{%
        \begin{tikzpicture}
            \node[shape=circle,draw=black] (Qubit 1) at (0,0) {Qubit 1};
            \node[shape=circle,draw=black] (Qubit 2) at (0,3) {Qubit 2};
            \node[shape=circle,draw=black] (Qubit 3) at (3,3) {Qubit 3};
            \node[shape=circle,draw=black] (Qubit 4) at (3,0) {Qubit 4};
        
            \path [-] (Qubit 2) edge node[left] {} (Qubit 4);
        \end{tikzpicture}
        }
        \caption{Toy Connectivity Graph}\label{fig:toy-connectivity}
        \vspace*{-10pt}
    \end{figure}

    Consider the connectivity graph $G$ shown in Figure \ref{fig:toy-connectivity}.
    This induces an accessible Lie algebra $\mathcal{H}$ as follows, where $a \in \{0,1,2,3\}$, $P^a$ is the $a$-th Pauli operator, $P^a_i$ is the $a$-th Pauli operator acting on qubit $i$, $\langle \_ \rangle_{\text{Lie}}$ denotes closure under the Lie bracket and $\imath$ is the imaginary unit:
    \begin{equation}
        \mathcal{H} := \langle \{ \imath P^a_i P^a_j : 
        (i,j) \in G, P^a \in \mathcal{P} \setminus \{i I_2\}\} \rangle_{\text{Lie}}.
    \end{equation}
    We want to compile the Hamiltonian
    \begin{multline}\label{eq:traceless-cnot}
        H_{\text{logical}} = 
        \log\left(CNOT/\left(\det(CNOT)^{1/2^2}\right)\right) - \\
        \Tr\left(\log\left(CNOT/\left(\det(CNOT)^{1/2^2}\right)\right)\right) I,
    \end{multline}
    where in this case $\log$ is the principal logarithm on the $[[4,2,2]]$ code.
    Note that the second summand in (\ref{eq:traceless-cnot}) exists to make the whole operator traceless while only modifying the global phase of its exponentiated form.
    The trick for an arbitrary Hamiltonian $H$ on a $d$-dimensional system is to take
    \begin{equation}
        H_{\text{traceless}} := H - \frac{\Tr(H)}{d} I.
    \end{equation}
    Given that $I$ commutes with $H$, one has
    \begin{multline}
        \exp(H_{\text{traceless}}) = \exp(H - \frac{\Tr(H)}{d} I) = \\
        \exp(H) \exp(- \frac{\Tr(H)}{d} I) \propto \exp(H),
    \end{multline}
    since
    \begin{equation}
        \exp(- \frac{\Tr(H)}{d} I) = \exp(\imath \theta) I
    \end{equation} for some $\theta \in [0, 2 \pi)$.
    At the same time, $H_{\text{traceless}}$ is manifestly traceless.
    
    The $[[4,2,2]]$ code has a Clifford encoding map represented in the binary symplectic space by
    \begin{equation}
        F_C = 
        \begin{bsmallmatrix}
            1&1&0&0&0&0&0&0\\
            0&1&0&1&0&0&0&0\\
            0&0&0&0&1&1&0&1\\
            1&0&1&1&1&1&1&1\\
            0&0&0&0&1&0&1&0\\
            0&0&0&0&0&0&1&1\\
            1&1&1&1&0&0&0&0\\
            0&0&0&0&1&1&1&1
        \end{bsmallmatrix}.
    \end{equation}
    The $2n \times 2n$ matrix $F_C$ is interpreted as follows: The single-qubit Pauli $I, X, Y, Z$ operators are identified with the standard basis vectors $(a_1, \ldots, a_n | b_1, \ldots, b_n)$ for $\mathbb{F}_2^{2n}$ via the map
    \begin{equation}
        (a_1, \ldots, a_n | b_1, \ldots, b_n) \mapsto \bigotimes_{i=1}^n \imath^{a_i b_i} X_i^{a_i} Z_i^{b_i}.
    \end{equation}
    Then, the rows of $F_C$ describe the action via conjugation of $C$ on the $2n$ standard $X$ and $Z$ Pauli operators on $n$ qubits.
    The matrix $F_C$ induces the following unitary representation in the computational basis (fixing the determinant to be $1$):
    \setcounter{MaxMatrixCols}{16}
    \begin{multline}
        C = \\
        \scalebox{0.7}{
        $\begin{bsmallmatrix}
             \frac{1}{\sqrt{2}} & 0 & \frac{1}{\sqrt{2}} & 0 & 0 & 0 & 0 & 0 & 0 & 0 & 0 & 0 & 0 & 0 & 0 & 0 \\
             0 & 0 & 0 & 0 & 0 & -\frac{\imath}{\sqrt{2}} & 0 & \frac{\imath}{\sqrt{2}} & 0 & 0 & 0 & 0 & 0 & 0 & 0 & 0 \\
             0 & 0 & 0 & 0 & 0 & 0 & 0 & 0 & 0 & 0 & 0 & 0 & 0 & \frac{\imath}{\sqrt{2}} & 0 & \frac{\imath}{\sqrt{2}} \\
             0 & 0 & 0 & 0 & 0 & 0 & 0 & 0 & -\frac{1}{\sqrt{2}} & 0 & \frac{1}{\sqrt{2}} & 0 & 0 & 0 & 0 & 0 \\
             0 & -\frac{\imath}{\sqrt{2}} & 0 & \frac{\imath}{\sqrt{2}} & 0 & 0 & 0 & 0 & 0 & 0 & 0 & 0 & 0 & 0 & 0 & 0 \\
             0 & 0 & 0 & 0 & \frac{1}{\sqrt{2}} & 0 & \frac{1}{\sqrt{2}} & 0 & 0 & 0 & 0 & 0 & 0 & 0 & 0 & 0 \\
             0 & 0 & 0 & 0 & 0 & 0 & 0 & 0 & 0 & 0 & 0 & 0 & -\frac{1}{\sqrt{2}} & 0 & \frac{1}{\sqrt{2}} & 0 \\
             0 & 0 & 0 & 0 & 0 & 0 & 0 & 0 & 0 & \frac{\imath}{\sqrt{2}} & 0 & \frac{\imath}{\sqrt{2}} & 0 & 0 & 0 & 0 \\
             0 & 0 & 0 & 0 & 0 & 0 & 0 & 0 & 0 & -\frac{\imath}{\sqrt{2}} & 0 & \frac{\imath}{\sqrt{2}} & 0 & 0 & 0 & 0 \\
             0 & 0 & 0 & 0 & 0 & 0 & 0 & 0 & 0 & 0 & 0 & 0 & \frac{1}{\sqrt{2}} & 0 & \frac{1}{\sqrt{2}} & 0 \\
             0 & 0 & 0 & 0 & -\frac{1}{\sqrt{2}} & 0 & \frac{1}{\sqrt{2}} & 0 & 0 & 0 & 0 & 0 & 0 & 0 & 0 & 0 \\
             0 & \frac{\imath}{\sqrt{2}} & 0 & \frac{\imath}{\sqrt{2}} & 0 & 0 & 0 & 0 & 0 & 0 & 0 & 0 & 0 & 0 & 0 & 0 \\
             0 & 0 & 0 & 0 & 0 & 0 & 0 & 0 & \frac{1}{\sqrt{2}} & 0 & \frac{1}{\sqrt{2}} & 0 & 0 & 0 & 0 & 0 \\
             0 & 0 & 0 & 0 & 0 & 0 & 0 & 0 & 0 & 0 & 0 & 0 & 0 & -\frac{\imath}{\sqrt{2}} & 0 & \frac{\imath}{\sqrt{2}} \\
             0 & 0 & 0 & 0 & 0 & \frac{\imath}{\sqrt{2}} & 0 & \frac{\imath}{\sqrt{2}} & 0 & 0 & 0 & 0 & 0 & 0 & 0 & 0 \\
             -\frac{1}{\sqrt{2}} & 0 & \frac{1}{\sqrt{2}} & 0 & 0 & 0 & 0 & 0 & 0 & 0 & 0 & 0 & 0 & 0 & 0 & 0 \\
        \end{bsmallmatrix}$.
        }
    \end{multline}
    The na{\"i}ve Hamiltonian implementation of $CNOT$ is then
    \begin{equation}
        H_{CNOT_{1 \to 2}} = C K \left(H_{\text{logical}} \oplus 0_{12} \right)K^\dagger C^\dagger.
    \end{equation}
    We find that the naïve implementation of CNOT on the $[[4,2,2]]$ code is not accessible.
    However, by applying Theorem~\ref{thm:projection-minimizer} we find that the SWAP Hamiltonian
    \begin{equation}
        H_{\text{SWAP} 2 \leftrightarrow 4} \propto X_2 X_4 + Y_2 Y_4 + Z_2 Z_4
    \end{equation}
    is indeed an accessible implementation of CNOT.
    
    \subsection{Non-Trivial Example}
    \begin{figure}[t]
    \centering
        \scalebox{0.9}{%
        \begin{tikzpicture}
            \node[shape=circle,draw=black] (Qubit 1) at (0,0) {Qubit 1};
            \node[shape=circle,draw=black] (Qubit 2) at (0,3) {Qubit 2};
            \node[shape=circle,draw=black] (Qubit 3) at (3,3) {Qubit 3};
            \node[shape=circle,draw=black] (Qubit 4) at (3,0) {Qubit 4};
        
            \path [-] (Qubit 1) edge node[left] {} (Qubit 3);
            \path [-] (Qubit 2) edge node[left] {} (Qubit 4);
            \path [-] (Qubit 1) edge node[left] {} (Qubit 2);
            \path [-] (Qubit 3) edge node[left] {} (Qubit 4);
        \end{tikzpicture}
        }
        \caption{Non-Trivial Connectivity Graph}\label{fig:connectivity}
        \vspace*{-10pt}
    \end{figure}
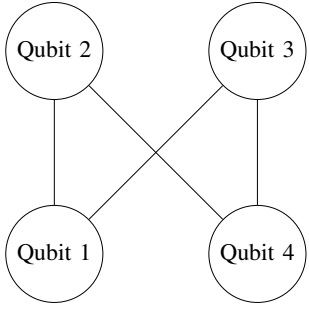
    
    Consider the connectivity graph $G$ shown in Figure \ref{fig:connectivity}.
    This induces an accessible Lie algebra $\mathcal{H}$ as before.
    In fact, the naïve implementation of CNOT is an element of $\mathcal{H}$.
    Instead, then, we turn to random implementations.
    We find that the implementation of CNOT engendered by a randomly generated stabilizer is usually not accessible.
    Yet, we find a 12-Pauli-term implementation of CNOT distinct from the SWAP Hamiltonian which is accessible.
    This shows that the pseudo-inverse solution to the least squares problem does not always recover the sparsest solution in the Pauli basis.
    Hence, this further motivates the regularized formulation discussed earlier.
    We plan to investigate it deeply in future work.

\section{Conclusion}
\label{sec:Conclusion}

    \subsection{Practical Considerations}
    \label{subsec:Practical-Considerations}
    
        We emphasize that from the experimentalist's point of view, it is more natural to work with Hamiltonians than unitaries as these are the basic operations available to them.
        In this sense, our decision to work in the Hamiltonian setting, rather than the unitary one, is practically, as well as theoretically, motivated.
        
        In an immediate sense, our primary least-squares method confers a benefit over na{\"i}ve compilation only when the accessible Lie algebra is a proper subset of the full special unitary Lie algebra.
        Thus, if the connectivity graph in question is connected with a vertex of degree exceeding 2, then our method will confer no obvious benefit~\cite[Theorem~I.1]{kökcü2024}.
        However, we argue that even if in this case, there may be benefits to our method achieved by excising edges of the connectivity graph.
        In some architectures with all-to-all connectivity, for example ion traps, long-distance connections may be realized by physically moving qubits, which can be relatively expensive in terms of runtime and coherence.
        Thus, it may be advantageous to first check if more connectivity-constrained solutions exist prior to making use of such long-range connections.
        More generally, it might be advantageous to use connectivity graphs to enforce different notions of \emph{easy} Hamiltonians to be implemented. 
        In the extreme case, for instance, we could enforce transversality by demanding that the connectivity graph be devoid of edges.
        Of course, our method could also be of use in this case given a satisfactory assignment of costs to elements of a basis for $\mathfrak{su}(2^n)$.

        In regard to runtime, our numerics suggest that our method scales exponentially in the number of physical qubits for a fixed number of logical qubits, though we note that we can solve a problem instance on the $[[5.1,3]]$ code in under $50$ milliseconds.
        However, we have yet to make use of \emph{parallel} methods for computing the Moore-Penrose Pseudoinverse which is the predominant step of our algorithm in terms of runtime.
        Such methods could drastically reduce the runtime necessary to perform our analysis computationally.
        
    \subsection{Future Work}
    \label{subsec:Future-Work}
    
    The absence of parallel methods for computing the Moore-Pensore Pseudoinverse here leaves open the possibility of dramatic time savings using parallel methods such as that implemented in CUDA. It would be productive to explore such methods in detail.

    In terms of applications, we view distributed qLDPC architectures as a natural candidate for our method. 
    Therein, connectivity is naturally limited and long-range connections are exorbitantly expensive.
    Thus, our method becomes a natural candidate for compilation.

    Also, K{\"o}kc{\"u} \emph{et al.}~\cite{kökcü2024} have already studied the Lie-algebraic properties of $\mathcal{H}$ in terms of connectivity graphs that generate them, whereas Aguilar \emph{et al.}~\cite{aguilar2024classificationpauliliealgebras} have reduced a classification problem for arbitrary Lie algebras generated by Pauli operators to a graph-theoretic one.
    It may be both theoretically and practically fruitful to continue this line of work by studying the properties of $\mathcal{H}$ in terms of graphs which generate it, in a suitable sense, with an eye towards the applicability of our method.
    Namely, which properties of the graph generating $\mathcal{H}$ guarantee that our method finds an accessible Hamiltonian for a given target and which properties can be exploited to speed up the computation of the Moore-Penrose Pseudoinverse, especially in a parallel fashion?
    It would be advantageous, for instance, if we were guaranteed that $(M \circ A)$ had some block structure beyond that discussed here.
    At the same time, investigating more sophisticated methods for regularization might be advantageous from a runtime perspective as well.

    Last, note that we began by fixing a labeling of physical qubits in our connectivity graph.
    This is a mathematical convenience rather than an operational necessity: one could apply any permutation of $n$ qubits to the existing connectivity graph and obtain an operationally valid connectivity graph.
    Yet, some permutations might be preferable to others insofar as they enable lower incompatibility via our procedure.
    For instance, if we exchange vertices $1$ and $2$ in Figure \ref{fig:connectivity}, our least-squares procedure no longer recovers an accessible implementation of CNOT.
    Thus, we have a co-design problem: how do we choose a permutation of qubits for which our procedure finds a solution of minimal incompatibility in the context of \emph{all possible permutations}?
    Resolving this question might augment the power of our current procedure.
    At first glance, we note that the Lie group of stabilizers has a rich analytic structure which we extensively make use of whereas the group of permutations on $n$ qubits is a strictly algebraic object.
    Thus, this co-design problem is one between objects which do not appear to be mathematically similar.

    \subsection{Summary}\label{subsec:Summary}
    In terms of encoded quantum computation, it suffices to implement any operator logically equivalent to some target, yet from a compilation perspective, certain choices may be preferable to others.
    Our novel contribution is making this intuition precise in the general setting of the special unitary group.
    In particular, we described how to reduce the problem of making a compilation-ideal choice to a least squares problem and provided a closed form and efficiently computable solution thereof. 
    We provided a time complexity bound on our method, and a worked example in which we discover a novel route for implementing logical $CNOT$ on the $[[4,2,2]]$ error-detecting code.

\section*{Acknowledgments}

We thank Iman Marvian for a helpful discussion related to traceless Hamiltonian generators.
We thank Jonathan Baker for helpful discussions related to weighted Hamiltonian costs and vertex permutations.
We used generative AI - namely Claude Opus 4.6 - to generate implementations of Algorithm \ref{alg:summary}, both as written and using the improved procedure of Section \ref{subsec:Improved-Time-Complexity}, in Python. The numerics generated by these implementations are found in Section \ref{sec:Example}.
This work was supported by the U.S. National Science Foundation Grant no. 2106189.

\bibliographystyle{IEEEtran}
\bibliography{refs.bib}

\end{document}